\documentclass[amsmath,11pt]{article}
\usepackage{times}
\usepackage{fullpage}
\usepackage{amsmath,amsfonts,amssymb,amsthm}
\usepackage{url}
\usepackage{algorithm}
\usepackage[noend]{algorithmic}

\newif\ifnotes\notestrue

\ifnotes

\newcommand{\dnote}[1]{{\bf (Daniel:} {#1}{\bf ) }}
\newcommand{\snote}[1]{{\bf (Santosh:} {#1}{\bf ) }}

\else

\newcommand{\dnote}[1]{}
\newcommand{\snote}[1]{}

\fi

\theoremstyle{plain}            
\newtheorem{theorem}{Theorem}[section]
\newtheorem{lemma}[theorem]{Lemma}
\newtheorem{corollary}[theorem]{Corollary}

\theoremstyle{definition}       

\theoremstyle{remark}           

\numberwithin{equation}{section}

\DeclareMathOperator{\poly}{poly}

\DeclareMathOperator*{\E}{E}

\newcommand{\R}{\ensuremath{\mathbb{R}}}

\newcommand{\pr}[2]{\langle{#1, #2}\rangle}

\newcommand{\abs}[1]{\lvert{#1}\rvert}

\newcommand{\set}[1]{\{{#1}\}}

\newcommand{\length}[1]{\lVert{#1}\rVert}

\def\tr{\mathrm{tr}}

\def\eps{\epsilon}

\def\vol{\mathrm{vol}}

\DeclareMathOperator*{\conv}{conv}

\title{Near-Optimal Deterministic Algorithms for Volume Computation and Lattice Problems via M-Ellipsoids}

\author{Daniel Dadush\thanks{School of Industrial and Systems Engineering, Georgia Tech. {\tt dndadush@gmail.com}}
\and
Santosh Vempala\thanks{School of Computer Science, Georgia Tech. {\tt vempala@gatech.edu}}}

\begin{document}
\date{April 2012}
\maketitle

\begin{abstract}
We give a deterministic $2^{O(n)}$ algorithm for computing an M-ellipsoid of a convex body, matching a known
lower bound. This has several interesting consequences including improved
deterministic algorithms for volume estimation of convex bodies and for the shortest and closest lattice vector problems under general norms.
\end{abstract}

\section{Introduction}

Ellipsoids have traditionally played an important role in the study of convex bodies. The classical Lowner-John ellipsoid, for instance, is the starting point for many interesting studies. To recall John's theorem, for any convex body $K$
in $\R^n$, there is an ellipsoid $E$ with centroid $x_0$ such that
\[
x_0 + E \subseteq K \subseteq x_0 + nE.
\]
In fact, this bound is achieved by the {\em maximum volume} ellipsoid contained in $K$.

Ellipsoids have also been critical to the design and analysis of efficient algorithms. The most notable example is the ellipsoid algorithm \cite{Shor77,YN76}
for linear \cite{Kh80} and convex optimization \cite{GLS}, which represents a frontier of polynomial-time solvability. For the basic
problems of sampling and integration in high dimension, the {\em inertial} ellipsoid defined by the covariance matrix of a distribution is an important
ingredient of efficient algorithms \cite{KLS97, LV2, Vempala10}. This ellipsoid also achieves the bounds of John's theorem for general convex bodies (for
centrally-symmetric convex bodies, the max-volume ellipsoid achieves the best possible sandwiching ratio of $\sqrt{n}$ while the inertial ellipsoid could still
have a ratio of $n$).

Another ellipsoid that has played a critical role in the development of modern convex geometry is the M-ellipsoid (Milman's ellipsoid). This object was
introduced by Milman as a tool to prove fundamental inequalities in convex geometry (see e.g., Chapter 7 of \cite{Pis89}). An M-ellipsoid $E$ of a convex body
$K$ has small {\em covering numbers} with respect to K. We let $N(A,B)$ denote the number of translations of $B$ required to cover $A$. Then, as shown by Milman,
every convex body $K$ has an ellipsoid $E$ for which $N(K,E)N(E,K)$ is bounded by $2^{O(n)}$. This is the best possible bound up to a constant in the exponent. In
contrast, the John ellipsoid can have this covering bound as high as $n^{\Omega(n)}$.  The existence of M-ellipsoids now has several proofs in the literature by
Milman \cite{M86}, multiple proofs by Pisier \cite{Pis89}, and most recently, by Klartag \cite{Klartag2006}.

The complexity of computing these ellipsoids is interesting for its own sake, but also due to several important consequences that we will discuss presently.
John ellipsoids are hard to compute, but their sandwiching bounds can be approximated deterministically to within $O(\sqrt{n})$ in polynomial time. Inertial
ellipsoids can be approximated to arbitrary accuracy by random sampling in polynomial time. Algorithms for M-ellipsoids have been considered only recently. The
proof of Klartag \cite{Klartag2006} gives a randomized polynomial-time algorithm \cite{DPV-SVP-11}. In \cite{DV12}, we give a deterministic $O(\log n)^n$ time
and $\poly(n)$-space algorithm to compute the $\ell$-ellipsoid of any convex body. The $\ell$-ellipsoid yields an approximation to the M-ellipsoid, where the
product of covering estimates is $O(\log n)^n$ instead of the best possible bound of $2^{O(n)}$. It has been open
to give a deterministic algorithm for constructing an M-ellipsoid that achieves optimal covering bounds. The extent to
which randomness is essential for efficiency is a very interesting question in general, and specifically for problems
on convex bodies where separations between randomized and deterministic complexity are known in the general oracle model \cite{BF87, DyerFK89}.
Here we address the question of deterministic M-ellipsoid construction
and consider its algorithmic consequences for volume estimation and for fundamental lattice problems, namely
the Shortest Vector Problem (SVP) and Closest Vector Problem (CVP).

The core new result of this paper is a deterministic $2^{O(n)}$ algorithm for computing an M-ellipsoid of a convex body in the oracle model \cite{GLS}.
Moreover, there is a $2^{\Omega(n)}$ lower bound for deterministic algorithms, so this is the best possible up to a constant in the exponent. We state this result formally, then proceed to its consequences and
a nearly matching lower bound.

\begin{theorem}\label{thm:det-M-ellipsoid}
There is a deterministic algorithm that, given any convex body $K \subset \R^n$  specified by a membership oracle, finds an ellipsoid $E$ such that $N(K,E)N(E,K) \le 2^{O(n)}$. The time complexity
of the algorithm (oracle calls and arithmetic operations) is $2^{O(n)}$ and its space complexity is polynomial in $n$.
\end{theorem}

The first consequence is for estimating the volume of a convex body.  This is an ancient problem that has lead to many insights in algorithmic techniques,
high-dimensional geometry and probability theory. One one hand, the problem can be solved for any convex body presented in the general membership oracle model
in randomized polynomial time to arbitrary accuracy \cite{DFK}.  On the other hand, the following lower bound (improving on \cite{E86}) shows that deterministic
algorithms cannot achieve such approximations.

\begin{theorem}\cite{BF87}\label{thm:vol-lb}
Suppose there is a deterministic algorithm that takes a convex body $K$ as input and outputs $A(K), B(K)$ such that $A(K) \le \vol(K) \le B(K)$ and makes at most $n^a$ calls to the membership oracle for $K$. Then there is some convex body $K$ for which
\[
\frac{B(K)}{A(K)} \le \left(\frac{c n}{a \log n}\right)^{n/2}
\]
where $c$ is an absolute constant.
\end{theorem}
In particular, this implies that even achieving a $2^{O(n)}$ approximation requires $2^{\Omega(n)}$ oracle calls. Now the volume of an M-ellipsoid $E$ of $K$ is
clearly within a factor of $2^{O(n)}$ of the volume of $K$, thus Theorem \ref{thm:det-M-ellipsoid} gives a $2^{O(n)}$ algorithm that achieves this
approximation. And, as claimed, we have a lower bound of $2^{\Omega(n)}$ for computing an M-ellipsoid deterministically. We state this corollary formally.

\begin{theorem}\label{thm:det-vol}
There is a deterministic algorithm of time complexity (oracle calls and arithmetic operations) $2^{O(n)}$ and polynomial
space complexity that estimates the volume of a convex body given by a membership oracle to within a factor of $2^{O(n)}$.
\end{theorem}

A natural question is whether this can be generalized to a trade-off between approximation
and complexity. Indeed the following result of Barany and Furedi \cite{BF88} gives a lower bound.

\begin{theorem}\cite{BF88}\label{thm:vol-tradeoff}
For any $0 \le \eps \le 1$, any deterministic algorithm that estimates the volume of any input convex body to within a
$(1+\eps)^n$ given only a membership oracle to the body, must make at least $\Omega(1/\eps)^n$ queries to the membership oracle.
\end{theorem}

We show that the M-ellipsoid algorithm can be extended to give an algorithm that essentially matches this best possible complexity
for centrally symmetric convex bodies.

\begin{theorem}\label{thm:optimal-vol}
For any $0 \le \eps \le 1$, there is a deterministic algorithm that computes a $(1+\eps)^n$ approximation of the
volume of a given centrally symmetric convex body in $O(1/\eps)^n$ time and polynomial space.
\end{theorem}

Next we turn to lattice problems.
For a convex body $K \subseteq \R^n$, s.t. $0 \in {\rm interior}(K)$, the gauge function of $K$ is
\[
\|x\|_K = \inf \set{s \geq 0: x \in sK}
\]
for $x \in \R^n$. For symmetric $K$ (i.e. $K=-K$), $\|\cdot\|_K$ is a usual norm on $\R^n$
(we shall refer to $\|\cdot\|_K$ as the norm induced by $K$ and specify asymmetric
whenever relevant).

In recent work, M-ellipsoids were shown to be useful for solving basic lattice problems \cite{DPV-SVP-11} of SVP and CVP.
The Shortest Vector Problem (SVP) is
stated as follows: given an $n$-dimensional lattice $L$ represented by a basis, and a norm defined by a convex body $K$, find a nonzero $v \in L$ such that
$\|v\|_K$ is minimized. In the Closest Vector Problem (CVP), in addition to a lattice and a norm, we are also given a query point $x$ in $\R^n$, and the goal is
to find a vector $v \in L$ that minimizes $\|x-v\|_K$. These problems are central to the geometry of numbers and have applications
to integer programming, factoring polynomials, cryptography, etc.
The fastest known algorithms for solving SVP under general norms, are $2^{O(n)}$ time randomized algorithms based on the AKS sieve
\cite{DBLP:conf/stoc/AjtaiKS01,DBLP:conf/fsttcs/ArvindJ08}. Finding deterministic algorithms of this complexity for both SVP and CVP has been an important open problem.

In fact, the AKS sieve uses an exponential amount of randomness. Improving on this, \cite{DPV-SVP-11} gave a $2^{O(n)}$ Las Vegas algorithm for general norm SVP
which uses only a polynomial amount of randomness. For CVP the complexity was $(2+\gamma)^{O(n)}$ assuming the minimum distance of the query point is at most
$\gamma$ times the length of the shortest vector. In subsequent work \cite{DV12},
we gave a deterministic $O(\log n)^n$ algorithm for the same results. In this paper, we
completely eliminate the randomness. In the statements below, we say that $K$ is {\em well-centered} if $\vol(K \cap -K) \ge 2^{-O(n)} \vol(K)$.
(every convex body is well-centered with respect to its centroid or a point sufficiently close to its centroid).

\begin{theorem}\label{thm:det-SVP}
Given a basis for a lattice $L$ and a well-centered norm $\|.\|_K$ specified by a convex body $K$ both in $\R^n$, the shortest vector in $L$ under the
norm $\|.\|_K$ can be found deterministically using $2^{O(n)}$ time and space.
\end{theorem}

\begin{theorem}\label{thm:det-CVP}
Given a basis for a lattice $L$, any well-centered $n$-dimensional convex body $K$ and a query point $x$ in $\R^n$, the closest vector in $L$ to $x$ in the norm
$\|.\|_K$ defined by $K$ can be computed  deterministically using $(2+\gamma)^{O(n)}$ time and space, provided that the minimum distance is at most $\gamma$
times the length of the shortest nonzero vector of $L$ under $\length{\cdot}_{K}$.
\end{theorem}

The approach in \cite{DPV-SVP-11} is to reduce the problem for general norms to the the Euclidean norm, or more specifically, to enumerating lattice points in ellipsoids.  We describe summarize the reduction in Section \ref{sec:SVP}.
In \cite{DPV-SVP-11}, the M-ellipsoid construction is a randomized polynomial-time algorithm based on the existence proof by Klartag \cite{Klartag2006}.
This approach is based on estimating a covariance matrix and seems inherently difficult to derandomize. In \cite{DV12}, we gave a deterministic algorithm based on computing an approximate minimum mean-width ellipsoid. For this approximation, we get that the covering bound is $N(K,E)N(E,K) = O(\log n)^n$, giving a deterministic algorithm of this complexity. Here we completely algorithmicize Milman's existence proof, to obtain the best possible deterministic complexity of $2^{O(n)}$.
By adjusting the parameters in the resulting algorithm to ``slow down" Milman's iteration, we get the optimal trade-off
between approximation and complexity for volume computation.

\section{Techniques from convex geometry}


\subsection{The Lewis ellipsoid}
Let $\alpha$ be a norm on $n \times n$ matrices. We define the dual norm $\alpha^*$ for any $S \in \R^{n \times n}$ as
\begin{equation}
\label{eq:dual-norm}
\alpha^*(S) = \sup \{\tr(SA) \, : \, A \in \R^{n\times n}, \alpha(A) \le 1\}.
\end{equation}
For a matrix $A \in \R^{n \times n}$, we denote its transpose by $A^T$, and its inverse (when it exists) by $A^{-1}$.

\begin{theorem}\label{thm:Lewis}\cite{Lewis79}
For any norm $\alpha$ on $\R^{n \times n}$, there is an invertible linear transformation $A \in \R^{n \times n}$ such that
\[
\alpha(A) = 1 \mbox{ and } \alpha^*(A^{-1}) = n.
\]
\end{theorem}

The proof of the above theorem is based on examining the properties of the optimal solution to the following mathematical program:
\begin{align}
\label{Lewis-prog}
\begin{split}
\max \det(A)&\\
\text{s.t.} & \\
	           A &\in \R^{n \times n} \\
                   \quad \alpha(A) & \leq 1
\end{split}
\end{align}
From here, showing that the optimal $A$ satisfies $\alpha^*(A^{-1})$ is a simple variational argument (reproduced in Lemma \ref{lem:approx-sdp}).

%

We will be interested in norms $\alpha$ of the following form. Let $K \subseteq \R^n$ denote a symmetric convex body with associated
norm $\|\cdot\|_K$, and let $\gamma_n$ denote the canonical Gaussian measure on $\R^n$. We define the $\ell$-norm with respect to $K$
for $A \in \R^{n \times n}$ as
\[
\ell_K(A) = \left(\int \|Ax\|^2_K d\gamma_n(x)\right)^{1/2}
\]
The $\ell$-norm was first studied and defined by Tomczak-Jaegermann and Figiel \cite{TjF79}.

The next crucial ingredient is a connection between the dual norm $\alpha^*$ defined above and the $\ell$-norm with respect to the polar $K^* = \set{x \in \R^n:
\pr{x}{y} \leq 1 ~ \forall y \in K}$, namely,
\[
\ell_{K^*}(A) = \left(\int \|Ax\|^2_{K^*} d\gamma_n(x)\right)^{1/2}.
\]

For two convex bodies $K,L \subseteq \R^n$ the Banach-Mazur distance between $K$ and $L$ is
\[
d_{BM}(K,L) = \inf \set{s: s \geq 1, TK \subseteq L-x \subseteq sTK, x \in \R^n,
T \in \R^{n \times n} \text{ invertible }}
\]

\begin{lemma}\label{lem:K-convexity-bound}\cite{Pis89}
For $A \in \R^{n \times n}$, we have that
\[
\ell_{K^*}(A^T) \le 4(1+\log d_{BM}(K,B_2^n)) \ell_K^*(A)
\]
\end{lemma}


\subsection{Covering numbers and volume estimates}

Let $B_2^n \subseteq \R^n$ denote the $n$-dimensional Euclidean ball. Recall
that $N(K,D)$ is the number of translates of $D$ required to cover $K$. The
following bounds for convex bodies $K,D \subset \R^n$ are classical. We use $c,C$ to
denote absolute constants here and later in the paper.

\begin{lemma}\label{lem:covering1}
For any two symmetric convex bodies $K,D$,
\[
\frac{\vol(K)}{\vol(K \cap D)} \le N(K,D) \le 3^n \frac{\vol(K)}{\vol(K \cap D)}.
\]
\end{lemma}



The next lemma is from \cite{Mil88}.

\begin{lemma}\label{lem:volume1}
Let $D \subseteq \alpha K$, $\alpha \geq 1$. Then,
\[
\vol(\conv{\{K,D\}}) \le 4 \alpha n N(D,K) \vol(K).
\]
\end{lemma}

The following are the Sudakov and dual Sudakov inequalities (see e.g., Section $6$ of \cite{GiaM01}).
\begin{lemma}[Sudakov Inequality]\label{lem:Sudakov}
For any $t > 0$, and invertible matrix $A \in \R^{n \times n}$
\[
N(K,tAB_2^n) \le e^{C\ell_{K^*}(A^{-T})^2/t^2}.
\]
\end{lemma}

\begin{lemma}[Dual Sudakov Inequality]\label{lem:DualSudakov}
For any $t > 0$, and $A \in \R^{n \times n}$
\[
N(AB_2^n, tK) \leq e^{C\ell_{K}(A)^2/t^2}.
\]
\end{lemma}

The following lemma gives a simple containment relationship (see e.g., \cite{DV12}).
\begin{lemma}\label{lem:ell-containment}
For any $A \in \R^{n \times n}$, $A$ invertible, we have that
\[
\frac{1}{\ell_{K^*}(A^{-1})}K \subseteq AB_2^n \subseteq \ell_K(A) K
\]
\end{lemma}
\begin{proof}
We first show that $E=AB_2^n \subseteq \ell_K(A) K$. Assume not, then there exists $x \in E$ such that $\|x\|_K = \sup_{y \in K^*} |\pr{y}{x}| > \ell_K(A)$. Now pick
$y \in K^*$ achieving $|\pr{y}{x}| = \|x\|_K$. Then we have that
\[
\ell_K(A) < |\pr{x}{y}| \leq \sup_{z \in AB_2^n} |\pr{z}{y}| = \sup_{z \in B_2^n} |\pr{z}{A^ty}|  = \|A^ty\|_2
\]
But now note that
\[
\ell_K(A) = \E[\|AX\|_K^2]^{\frac{1}{2}} \geq \E[|\pr{y}{AX}|^2]^{\frac{1}{2}} = \|A^ty\|_2
\]
a clear contradiction. Therefore $AB_2^n \subseteq \ell_K(A) K$ as needed. Now applying the same argument on $E^* = A^{-1}B_2^n$ and $K^*$, we get that
$E^* \subseteq \ell_K(A^{-1})K^*$. From here via duality, we get that
\[
\frac{1}{\ell_{K^*}(A^{-1})}K = (\ell_{K^*}(A^{-1})K^*)^* \subseteq  (A^{-1}B_2^n)^* = AB_2^n
\]
as needed.
\end{proof}

\section{Algorithm for computing an M-ellipsoid}

In this section, we present the algorithm for computing an M-ellipsoid of an arbitrary convex body given in the oracle model. We first observe that it suffices
to give an algorithm for centrally symmetric $K$. For a general convex body $K$, we may replace $K$ by the difference body $K-K$ (which is symmetric). An
$M$-ellipsoid for $K-K$ remains one for $K$, as the covering estimates changes by at most a $2^{O(n)}$ factor. To see this, note that for any ellipsoid $E$
we have that $N(K,E) \leq N(K-K,E)$ and that
\[
N(E,K) \leq N(E,K-K)N(K-K,K) \leq N(E,K-K)2^{O(n)},
\]
where the last inequality follows from the Rogers-Shephard inequality~\cite{RS57}, i.e. $\vol(K-K) \leq 4^n \vol(K)$.

Our algorithm has two main components: a subroutine to compute an approximate Lewis ellipsoid for a norm given by a convex body, and an implementation of
the iteration that makes this ellipsoid converge to an M-ellipsoid of the original convex body.

\subsection{Approximating the $\ell$-norm}

Our approximation of the $\ell_K$ norm is as follows:
\[
\tilde{\ell}_K(A)=\sum_{x \in \{-1,1\}^n} \frac{1}{2^n} \|Ax\|_K.
\]

The next lemma is essentially folklore, we give a known proof here.
\begin{lemma}\label{easy-ell-approx}
For a symmetric convex body $K$ and any $A \in \R^{n \times n}$, we have
\[
\ell_K(A) \le 4\sqrt{\frac{\pi}{2}}(1+\log d_{BM}(K,B_2^n)) \tilde{\ell}_K(A).
\]
\end{lemma}
\begin{proof}
Let $g_1,\dots,g_n$ denote i.i.d. $N(0,1)$ Gaussians, let $u_1,\dots,u_n$ denote i.i.d. uniform $\set{-1,1}$ random variables and let $A_1,\dots,A_n \in \R^n$ denote
the columns of $A$. Then we have that
\begin{align*}
\ell_K(A) &\leq 4(1+\log d_{BM}(K,B_2^n)) \sup \left\{\sum_i \pr{A_i}{y_i}: \E[\|\sum_i g_iy_i\|_{K^*}^2]^{\frac{1}{2}} \leq 1 \right\} \\
          &\leq 4\sqrt{\frac{\pi}{2}} (1+\log d_{BM}(K,B_2^n)) ~ \sup \left\{\sum_i \pr{A_i}{y_i}: \E[\|\sum_i u_iy_i\|_{K^*}^2]^{\frac{1}{2}} \leq 1\right\} \\
          &\leq 4\sqrt{\frac{\pi}{2}} (1+\log d_{BM}(K,B_2^n)) ~ \E[\|\sum_i u_iA_i\|_K^2]^{\frac{1}{2}} =  4\sqrt{\frac{\pi}{2}} (1+\log d_{BM}(K,B_2^n)) ~ \tilde{\ell}_K(A)
\end{align*}
Here, the first inequality follows by Lemma \ref{lem:K-convexity-bound}. The next inequality follows from the classical comparison
$\E[f(u_1,\dots,u_n)] \leq \E[f(\sqrt{\frac{\pi}{2}}g_1,\dots,\sqrt{\frac{\pi}{2}}g_n)]$ for any convex function $f:\R^n \rightarrow \R$, and setting $f(x_1,\dots,x_n)
= \|\sum_i x_iy_i\|_C^2$. The last inequality follows from the following weak duality relation:
\begin{align*}
\sum_i \pr{A_i}{y_i} &= \E[\pr{\sum_i u_iA_i}{\sum_j u_jy_j}] \leq \E[\|\sum_i u_iA_i\|_K \|\sum_j u_jy_j\|_{K^*}] \\
                     &\leq \E[\|\sum_i u_iA_i\|_K^2]^{\frac{1}{2}} \E[\|\sum_j y_ju_j\|_{K^*}^2]^{\frac{1}{2}} \leq \ell_K(A) \text{.}
\end{align*}
\end{proof}

The next lemma is a strengthening due to Pisier, using Proposition 8 from \cite{Pisier81}. While it is not critical 
for our results (the difference is only in absolute constants), we use this stronger bound in our analysis.
\begin{lemma}\label{lem:ell-approx}
For a symmetric convex body $K$ and any $A \in \R^{n \times n}$, we have
\[
\frac{1}{\sqrt{\frac{\pi}{2}}} \tilde{\ell}_K(A) \leq \ell_K(A) \le c_1 \tilde{\ell}_K(A) \sqrt{1+\log d_{BM}(K,B_2^n)}
\]
where $c_0,c_1$ are absolute constants. Furthermore, by duality, we get that
\[
\frac{1}{c_1 \sqrt{1+\log d_{BM}(K,B_2^n)}} \tilde{\ell}^*_K(A) \leq \ell_K^*(A) \leq  \sqrt{\frac{\pi}{2}} \tilde{\ell}^*(A) \text{.}
\]
\end{lemma}

\subsection{A convex program}

To compute the approximate $\ell$-ellipsoid we use the following convex program:
\begin{align}\label{sdp}
\begin{split}
\max \det(A)^{\frac{1}{n}} & \\
\text{s.t.} & \\
            \quad A &\succeq 0 \\
            \quad \tilde{\ell}_K(A) & \leq 1
\end{split}
\end{align}

Here the main thing we change is that we replace the $\ell$-norm with $\tilde{\ell}_K$. This will suffice for
our purposes. We optimize over only positive semidefinite matrices (unlike Lewis' program \ref{Lewis-prog}).
This enables us to ensure convexity of
program while maintaining the desired properties for the optimal solution. For convenience we use $\det(.)^{1/n}$
as the objective function and clearly this makes no essential difference.

\subsection{Main algorithm}

\begin{figure}[ht]
\fbox{\parbox{\textwidth}{

\noindent
{\bf M-ellipsoid.}
\begin{enumerate}
\item Let $K_1 = K$ and $T = \log^* n$
\item For $i=1 \ldots T-1$,
\begin{enumerate}
\item Compute an approximate $\ell$-ellipsoid of $K_i$ using the convex program (\ref{sdp}) to get an approximately
optimal transformation $A_i$ (the corresponding ellipsoid is $A_iB_2^n$).
\item Set
\[
r_{in} =  \frac{\sqrt{n}}{\log^{(i)}(n) \tilde{\ell}_{K_i}(A_i)} \mbox{ and }
r_{out} = \log^{(i)}(n)\frac{\tilde{\ell}_{K_i^*}(A_i^{-1})}{\sqrt{n}}.
\]
\item Define
\[
K_{i+1} = \mbox{conv}\{K_i \cap r_{out} A_iB_2^n, r_{in} A_iB_2^n \}.
\]
\end{enumerate}
\item Output $E = \frac{\sqrt{n}}{\tilde{\ell}_{K_{T-1}}(A_{T-1})}A_{T-1} B_2^n$ as the M-ellipsoid.
\end{enumerate}
}}
\caption{The M-Ellipsoid Algorithm}
\label{fig:algo}
\end{figure}

Given a convex body $K$, we put it in approximate John position using the Ellipsoid algorithm in polynomial time \cite{GLS},
so that $B_2^n \subseteq K \subseteq nB_2^n$. We then use the above procedure, which is essentially an algorithmic version of Milman's proof of the existence of
$M$-ellipsoids. In the description below, by $\log^{(i)}n$ we mean the $i$'th iterated logarithm, i.e., $\log^{(1)} n = 1, \log^{(2)} n = \log\log n$ and so on.

\section{Analysis}

We note that the time complexity of the algorithm is bounded by $\mbox{poly}(n)2^{O(n)}$ and the space complexity is polynomial in $n$. In fact, the only step
that takes exponential time is the evaluation of the $\ell$-norm constraint of the SDP. This evaluation happens a polynomial number of times.
The rest of computation involves applying the ellipsoid algorithm and computing oracles for successive bodies (for $K_{i+1}$ given an oracle
for $K_i$), both of which are fairly straightforward \cite{GLS}. In particular, we build an oracle
for the intersection of two convex bodies given by oracles and for the convex hull of two convex bodies given by oracles.
The oracle for a body consists of a membership test and a bound on the ratio between two balls that sandwich the body.
Our analysis below provides sandwiching bounds and the complexity of the oracle grows as $n^{O(i)}$ in the $i$'th iteration,
for a maximum of $n^{O(\log^* n)} = 2^{o(n)}$.

We begin by showing that Lewis's bound (Theorem \ref{thm:Lewis}) is robust to approximation
and works when restricted to positive semi-definite transformations.
This allows us to establish the desired properties for approximate optimizers of the convex program (\ref{sdp}).

\begin{lemma}\label{lem:approx-sdp}
Let $K$ be such that $B_2^n \subseteq K \subseteq nB_2^n$ and
$A \in \R^{n \times n}$, be a $(1-\eps)$-approximate optimizer for the convex program (\ref{sdp}),
i.e. $\det(A)^{\frac{1}{n}} \geq (1-\eps) OPT$. Then for $\eps \leq 1/36n^4$, we have that
\[
\tilde{\ell}_K(A)\tilde{\ell}_K^*(A^{-1}) \le n(1+ 6n^2\sqrt{\eps}) \le 2n.
\]
\end{lemma}
\begin{proof}
For simplicity of notation, we write $\tilde{\ell}_K(T)$ as $\alpha(T)$ for $T \in \R^{n \times n}$.  Take $T \in \R^{n \times n}$ (not
necessarily positive semidefinite) satisfying $\alpha(T) \leq 1$. Let $\|T\|_F = \sqrt{\sum_{i,j} T_{ij}^2}$ denote the frobenius norm of
$T$, and $\|T\|_2 = \sup_{x \in B_2^n} \|Tx\|_2$ denote the operator norm of $T$.

\paragraph{Claim: $\alpha(T) \le \|T\|_F \leq n\alpha(T)$.}
\begin{proof}
Let $U$ denote a uniform vector in $\set{-1,1}^n$. Since $\frac{1}{n}\|x\|_2 \leq \|x\|_K$ for any $x \in \R^n$, we have that
\[
\alpha(T) = \E[\|UT\|_K^2]^{\frac{1}{2}} \geq \frac{1}{n} \E[\|UT\|_2^2]^{\frac{1}{2}} = \frac{1}{n} \|T\|_F.
\]
Now using the inequality $\|x\|_K \leq \|x\|_2$ for $x \in \R^n$, a similar argument yields $\alpha(T) \leq \|T\|_F$.
\end{proof}

First note that $I_n/\alpha(I_n)$ is a feasible solution to (\ref{sdp}) satisfying
\[
\det(\frac{I_n}{\alpha(I_n)})^{\frac{1}{n}} = \frac{1}{\alpha(I_n)} \geq \frac{1}{\|I_n\|_F} = \frac{1}{\sqrt{n}}.
\]
Let $A_{OPT} \succeq 0$ denote the optimal solution to (\ref{sdp}).
Since $\det(A_{OPT}) \geq \frac{1}{\sqrt{n}}$, we clearly have that $A_{OPT} \succ 0$. Therefore for $\delta > 0$ small enough we have that
$A_{OPT} + \delta T \succeq 0$. From this, we see that $(A_{OPT}+\delta T)/\alpha(A_{OPT}+\delta T)$ is also feasible
for (\ref{sdp}) as $\alpha((A_{OPT}+\delta T)/\alpha(A_{OPT}+\delta T)) = 1$. Since $A_{OPT}$ is the optimal solution, we have that
\[
\det\left(\frac{A_{OPT}+\delta T}{\alpha(A_{OPT}+\delta T)}\right)^{\frac{1}{n}} \leq \det(A_{OPT})^{\frac{1}{n}}.
\]
Rewriting this and using the triangle inequality,
\begin{align*}
\det(A_{OPT}+\delta T)^{\frac{1}{n}} &\leq \det(A_{OPT})^{\frac{1}{n}} \alpha(A_{OPT}+\delta T) \leq \det(A_{OPT})^{\frac{1}{n}} (\alpha(A_{OPT}) + \delta \alpha(T)) \\
                                     &\leq \det(A_{OPT})^{\frac{1}{n}} (1 + \delta).
\end{align*}
Dividing by $\det(A_{OPT})^{\frac{1}{n}}$ on both sides, we get that
\begin{equation}
\det(I_n + \delta A_{OPT}^{-1}T)^{\frac{1}{n}} \leq 1+\delta.
\end{equation}
Since both sides are equal at $\delta = 0$, we must have the same inequality for the derivatives with respect to $\delta$ at $0$.
This yields
\begin{equation}
\label{eq:apr-sdp-1}
\frac{1}{n} \tr(A_{OPT}^{-1}T) \leq 1 \Leftrightarrow \tr(A_{OPT}^{-1}T) \leq n
\end{equation}
Up to this point the proof is essentially the same as Lewis' proof of Theorem \ref{thm:Lewis}. We now depart from
that proof to account for approximately optimal solutions.

\paragraph{Claim: $\|A^{-1}_{OPT}\|_2 \leq n$.}
\begin{proof}
Let $\sigma$ denote the largest eigenvalue of $A_{OPT}^{-1}$ and $v \in \R^n$ be an associated unit eigenvector.
Since $A_{OPT} \succ 0$, we
have that $A_{OPT}^{-1} \succ 0$, and hence $\sigma = \|A^{-1}\|_2$. Now note that $A_{OPT} + \delta v v^T \succ 0$ for any $\delta \geq 0$,
and that $\alpha(v v^T) \leq \|v v^T\|_F = \|v\|_2^2 = 1$. Therefore by Equation \eqref{eq:apr-sdp-1}, we have that
\[
n \geq \tr(A^{-1}(v v^T)) = \tr(\sigma v v^T) = \sigma
\]
as needed.
\end{proof}

\paragraph{Claim: $A^{-1} \preceq (1+6\sqrt{n \eps})A_{OPT}^{-1}$.}
\begin{proof}
Since $A$ is $(1-\eps)$-approximate maximizer to (\ref{sdp}) we have
that
\[
\det(A)^{\frac{1}{n}} \geq (1-\eps) \det(A_{OPT})^{\frac{1}{n}} \Rightarrow \det(A) \geq (1-n\eps) \det(A_{OPT})
\]
We begin by proving by proving $A \succeq (1-3\sqrt{n \eps})A_{OPT}$. Now note that
\[
A \succeq (1-3\sqrt{n \eps})A_{OPT} ~~ \Leftrightarrow ~~ A_{OPT}^{-\frac{1}{2}} A A_{OPT}^{-\frac{1}{2}} \succeq (1-3\sqrt{n \eps})I_n
\]
Hence letting $B = A_{OPT}^{-\frac{1}{2}} A A_{OPT}^{-\frac{1}{2}}$, it suffices to show that $B \succeq (1-3\sqrt{n \eps})I_n$. From here,
we note that $1 \geq \det(B) = \det(A) / \det(A_{OPT}) \geq (1-n\eps)$. Now from Equation \eqref{eq:apr-sdp-1}, we have that
\[
\tr(B) = \tr(A_{OPT}^{-\frac{1}{2}} A A_{OPT}^{-\frac{1}{2}}) = \tr(A_{OPT}^{-1} A) \leq n
\]
Let $\sigma_1,\dots,\sigma_n$ denote the eigen values of $B$ in non-increasing order. We first note that $\sigma_n \leq 1$ since otherwise
\[
\det(B) = \prod_{i=1}^n \sigma_i \geq \sigma_n^n > 1
\]
a contradiction. Furthermore, since $B \succ 0$, we have that $0 < \sigma_n \leq 1$.
So we may write $\sigma_n = 1-\eps_0$, for $1 > \eps_0 \geq 0$.
Now since $\sum_{i=1}^n \sigma_i = \tr(B) \le n$,
by the arithmetic mean - geometric mean inequality we have that
\[
\det(B) = \sigma_n \prod_{i=1}^{n-1} \sigma_i = (1-\eps_0) \prod_{i=1}^{n-1} \sigma_i \leq
(1-\eps_0)\left(\frac{\sum_{i=1}^{n-1}\sigma_i}{n-1}\right)^{n-1} \le (1-\eps_0)(1+\frac{\eps_0}{n-1})^{n-1}
\]
Using the inequality $1+x \leq e^x \leq 1+x+\frac{e-1}{2}x^2$ for $x \in [-1,1]$, we get that
\begin{align*}
(1-\eps_0)(1+\frac{\eps_0}{n-1})^{n-1} &\leq (1-\eps_0)e^{\eps_0} \leq (1-\eps_0)(1+\eps_0+ \frac{e-1}{2} \eps_0^2) \\
                                       &= 1 - \frac{3-e}{2} \eps_0^2 - \frac{e-1}{2} \eps_0^3 \leq 1 - \frac{3-e}{2} \eps_0^2
\end{align*}
From this we get that
\[
1 - \frac{3-e}{2} \eps_0^2 \geq \det(B) \geq (1-n\eps)  ~~\Rightarrow~~ \eps_0 \leq \sqrt{\frac{2}{3-e} n \eps} \leq 3 \sqrt{n \eps}
\]
Therefore $\sigma_n = 1-\eps_0 \geq 1- 3 \sqrt{n \eps} \Rightarrow B \succeq (1-3\sqrt{n \eps})I_n \Rightarrow A \succeq (1-3\sqrt{n \eps})A_{OPT}$ as needed.
From here we get that
\[
A^{-1} \preceq \left(\frac{1}{1-3\sqrt{n \eps}}\right)A_{OPT}^{-1} \preceq (1+6\sqrt{n \eps})A_{OPT}^{-1}
\]
for $\eps \leq 1/36n$, proving the claim.
\end{proof}

Now take $T \in \R^{n \times n}$ satisfying $\alpha(T) \leq 1$. By the first claim, we note that $\|T\|_F \leq n \alpha(T) \leq n$. Now by
Equation \eqref{eq:apr-sdp-1}, we have that
\[
\tr(A^{-1}T) = \tr(A^{-1}_{OPT} T) + \tr((A^{-1}-A^{-1}_{OPT}) T) \leq n + \|A^{-1}-A^{-1}_{OPT}\|_F \|T\|_F \leq n + n \|A^{-1}-A^{-1}_{OPT}\|_F
\]
We bound the second term using the previous claim.
Since $A^{-1} \preceq (1+6\sqrt{n \eps})A^{-1}_{OPT}$, we have that $A^{-1}-A^{-1}_{OPT} \preceq 6\sqrt{n \eps}A^{-1}_{OPT}$, and hence
\[
\|A^{-1}-A^{-1}_{OPT}\|_F \leq \sqrt{n}\|A^{-1}-A^{-1}_{OPT}\|_2 \leq 6 n \sqrt{\eps} \|A^{-1}_{OPT}\|_2 \leq 6 n^2 \sqrt{\eps}
\]
Using this bound, we get
\[
\tr(A^{-1}T) \leq n + 6 n^3 \sqrt{\eps} = n(1 + 6n^2\sqrt{\eps})
\]
for any $T \in \R^{n \times n}$ satisfying $\alpha(T) \leq 1$. Thus we get that
$\alpha^*(A^{-1}) \leq n\left(1+6n^2\sqrt{\eps}\right)$. Together with the constraint $\alpha(A) \leq 1$,
the conclusion of the lemma follows.
\end{proof}


\begin{theorem}\label{thm:approx-Lewis}
Let $A$ be a $(1-\eps)$-approximate optimizer to the convex program (\ref{sdp}) for $\eps \le 1/(36n^4)$. Then
\[
\ell_K(A)\ell_{K^*}(A^{-1}) \le Cn \log^{\frac{3}{2}} d_{BM}(K,B_2^n).
\]
for an absolute constant $C > 0$.
\end{theorem}
\begin{proof}
Using Lemma \ref{lem:approx-sdp}, we have that
\[
\tilde{\ell}_K(A)\tilde{\ell}_K^*(A^{-1}) \le 2n.
\]
Next we use the approximation property (Lemma \ref{lem:ell-approx}) of $\tilde{\ell}_K$ to derive that
\[
\ell_K(A)\ell_K^*(A^{-1}) \le C n \sqrt{\log d_{BM}(K,B_2^n))}.
\]
Finally, noting that $A^{-T} = A^{-1}$ (by symmetry of $A$), we apply Lemma \ref{lem:K-convexity-bound} to infer that
\[
\ell_{K^*}(A^{-1}) \le C \ell_{K}^*(A^{-1}) \log d_{BM}(K,B_2^n),
\]
which completes the proof.
\end{proof}

Next we turn to proving that the algorithm produces an M-ellipsoid. While the analysis follows the existence proof to
a large extent, we need to handle the various approximations incurred.

To aid in the analysis of Algorithm \ref{fig:algo} on input $K \subseteq \R^n$, we make some additional definitions.
Let $a_i = \log^{(i)} n$ and $T =\log^* n$.
Let $K_1,\dots,K_T$ and $A_1,\dots,A_T$ denote the sequence of bodies and transformations generated by the algorithm.
Set
$K^{out}_1 = K^{in}_1 = K$, and for $1 \leq i \leq T-1$ define
\[
K^{in}_{i+1} = {\rm conv}\set{K^{in}_i, r_{in}^i A_iB_2^n} \quad K^{out}_{i+1} = K^{out}_i \cap r_{out}^i A_iB_2^n
\]
where $r_{in}^i, r_{out}^i$ are defined as $r_{in}, r_{out}$ in the $i$'th iteration of the main loop in Algorithm \ref{fig:algo}.

By construction, we have the relations
\[
K \subseteq K^{in}_1 \subseteq \dots \subseteq K^{in}_T,
\quad \quad K \supseteq K^{out}_1 \supseteq \dots \supseteq K^{out}_T,
\quad \quad K^{out}_i \subseteq K_i \subseteq K^{in}_i ~~\forall i \in [T]
\]

The proof of the main theorem will be based on the following inductive lemmas which quantify the properties of the sequences of bodies defined above.

\begin{lemma}\label{lem:it-banach-mazur}
$\forall i \in [T]$, we have that $d_{BM}(K_i,B_2^n) \leq C (\log^{(i-1)}n)^{\frac{7}{2}}$.
\end{lemma}
\begin{proof}
For the base case, we have that $d_{BM}(K_1,B_2^n) \leq \sqrt{n} \leq C n^{\frac{7}{2}}$ for any constant $C \geq 1$.

For the general case, by construction of $K_{i+1}$ we have that
\[
r^i_{in}A_iB_2^n \subseteq K_{i+1} \subseteq r^i_{out}A_iB_2^n.
\]
Therefore,
\begin{eqnarray*}
d_{BM}(K_{i+1},B_2^n) &\le& r^i_{out}/r^i_{in} \\
&=& a_i^2 \tilde{\ell}_{K_i^*}(A_i^{-1}) \tilde{\ell}_{K_i}(A_i) /n \\
&\le& C_1 a_i^2 \ell_{K_i^*}(A^{-1})\ell_{K_i}(A_i)/n \quad \left(\text{by Lemma \ref{lem:ell-approx}}\right) \\
&\le& C_1 (\log^{(i)} n)^2 (\log d_{BM}(K_{i},B_2^n))^\frac{3}{2}. \quad \left(\text{by Lemma \ref{thm:approx-Lewis}}\right)
\end{eqnarray*}
Using the fact that $\log^{(i)}(n) \geq 1$, $\forall i \in [T-1]$, a direct computation shows that the above recurrence equation implies the existence of a constant
$C > 1$ (depending only on $C_1$) such that the stated bound on $d_{BM}(K_{i+1},B_2^n)$ holds.
\end{proof}

\begin{lemma}\label{lem:iteration}
For $i \in [T-1]$, we have that
\[
\max \left\{ \frac{\vol(K^{out}_i)}{\vol(K^{out}_{i+1})}, \frac{\vol(K^{in}_{i+1})}{\vol(K^{in}_i)}\right\} \le e^{Cn/\log^{(i)} n}
\]
\end{lemma}
\begin{proof}

By Lemma \ref{lem:covering1}, the fact that $K^{out}_i \subseteq K_i$, Lemma \ref{lem:Sudakov}, Lemma \ref{lem:ell-approx} and Lemma \ref{lem:it-banach-mazur}, we have that
\begin{align*}
\frac{\vol(K^{out}_i)}{\vol(K^{out}_{i+1})} &\le N(K^{out}_i, r^i_{out}A_iB_2^n)
                                            \le N(K_i, r^i_{out}A_iB_2^n) \\
                                            &\le e^{C (\ell_{K_i^*}(A_i^{-1})/r_{out}^i)^2}
                                            = e^{C n \ell_{K_i^*}(A_i^{-1})^2/(a_i \tilde{\ell}_{K_i^*}(A^{-1}))^2} \\
                                            &\le e^{C n \log(d_{BM}(K_i^*,B_2^n))/a_i^2}
                                            \le e^{C n/\log^{(i)}n}
\end{align*}
By Lemma \ref{lem:ell-containment}, \ref{lem:ell-approx} and \ref{lem:it-banach-mazur}, we see that
\[
r^i_{in}A_iB_2^n \subseteq r^i_{in}\ell_{K^{in}_i}(A_i) K^{in}_i \subseteq r^i_{in} \ell_{K_i}(A_i) K^{in}_i \subseteq C_1 \sqrt{n} K^{in}_i.
\]
Next by Lemma \ref{lem:volume1}, the fact that $K_i \subseteq K^{in}_i$, Lemma \ref{lem:DualSudakov}, Lemma \ref{lem:ell-approx} and Lemma \ref{lem:it-banach-mazur}, we have that
\begin{align*}
\frac{\vol(K^{in}_{i+1})}{\vol(K^{in}_i)} &\le C_14n^{\frac{3}{2}} N(r^i_{in}A_iB_2^n, K^{in}_i)
\le C_1n^{\frac{3}{2}} N(r^i_{in}A_iB_2^n, K_i) \\
&\le C_1n^{\frac{3}{2}} e^{C(\ell_{K_i}(A_i)r_{in}^i)^2} = C_1n^{\frac{3}{2}} e^{Cn \ell_{K_i}(A_i)^2/(a_i \tilde{\ell}_{K_i}(A_i))^2} \\
                                                &\le C_1n^{\frac{3}{2}} e^{Cn \log(d_{BM}(K_i,B_2^n))/a_i^2} \le C_1n^{\frac{3}{2}} e^{C n(1/\log^{(i)}n)}
                                                 \le e^{C n/\log^{(i)}n}
\end{align*}

\end{proof}

We are now ready to complete the proof.

\begin{proof}(of Theorem \ref{thm:det-M-ellipsoid}.)
By construction of $K_T$, we note that
\[
r^{T-1}_{in} A_{T-1} B_2^n \subseteq K_T \subseteq r^{T-1}_{out} A_{T-1} B_2^n
\]
where by Lemma \ref{lem:it-banach-mazur} we have that $r^{T-1}_{out}/r^{T-1}_{in} = O(1)$. Therefore the returned ellipsoid $E = \frac{\sqrt{n}}{\tilde{\ell}_{K_{T-1}(A_{T-1})}} A_{T-1}B_2^n$ (last line of Algorithm \ref{fig:algo}) satisfies that
\[
\frac{1}{C} E \subseteq K_T \subseteq C E
\]
for an absolute constant $C \geq 1$. Next by Lemma \ref{lem:covering1}, we have that
\[
N(K,E), N(E,K) \leq 3^n \frac{\max \set{ \vol(K), \vol(E) } }{ \vol(K \cap E) }
\]
Now we see that
\[
K \subseteq K^{in}_T \subseteq C K^{in}_T \quad \quad E \subseteq C K_T \subseteq C K^{in}_T,
\]
and that
\[
K \supseteq \frac{1}{C} K^{out}_T \quad \quad E \supseteq \frac{1}{C} K_T \supseteq \frac{1}{C} K^{out}_T \text{.}
\]
Therefore,
\[
\frac{ \max \set{ \vol(K), \vol(E) } }{\vol(K \cap E)} \leq C^{2n} \frac{\vol(K^{in}_T)}{\vol(K^{out}_T)} \text{.}
\]
Finally, by Lemma \ref{lem:iteration} we have that
\begin{align*}
\frac{\vol(K^{in}_T)}{\vol(K^{out}_T)} = \prod_{i=1}^{T-1} \frac{\vol(K^{in}_{i+1})}{\vol(K^{in}_{i})} ~ \frac{\vol(K^{out}_i)}{\vol(K^{out}_{i+1})}
                                    \leq \prod_{i=1}^{T-1} e^{2Cn/\log^{(i)}n} = 2^{O(n)} \text{.}
\end{align*}
Combining the above inequalities yields the desired guarantee on the algorithm.
The time complexity is $2^{O(n)}$, dominated by the time to evaluate the $\tilde{\ell}_K$-norm. The space is polynomial
since all we need to maintain are efficient oracles for the successive bodies $K_i$, which can be done space-efficiently
for the operations of intersection and convex hull used in the algorithm \cite{GLS}.
\end{proof}

\section{An asymptotically optimal volume algorithm}\label{sec:volume}

As noted in the introduction, the result of Theorem \ref{thm:det-vol}, a deterministic $2^{O(n)}$-approximation for volume,
 follows directly from Theorem \ref{thm:det-M-ellipsoid}.
In this section, we show how to modify our M-ellipsoid algorithm (based on Milman's iteration) to match this lower bound
algorithmically.

In the M-ellipsoid algorithm of the previous section, we construct a series of convex bodies $K_0 = K, K_1, \ldots, K_T$
such that the covering numbers $N(K,K_T)$ and $N(K_T,K)$ are bounded by $2^{O(n)}$ and the final body $K_T$ has
$d_{BM}(K_T, B_2^n) < C$ for some constant $C$. Our modification will construct a similar sequence of bodies, but rather
than bounding covering numbers, we will ensure that
\[
e^{-C\eps n} \vol(K) \le \vol(K_T) \le e^{C\eps n} \vol(K)
\]
and
\[
d_{BM}(K_T, B_2^n) \le C \frac{\ln(1/\eps)^{\frac{5}{2}}}{\eps^2}.
\]
Then we approximate the volume of $K_T$ by finding an approximate $\ell$-ellipsoid $E$ for it, and covering it with
translations of a maximal parallelopiped that fits in $\eps E$. Since this covering will consist of disjoint
parallelopipeds, and their union will be contained in $K_T+\eps E \subseteq (1+\eps)K_T$,
we get the desired approximation. Here is the precise algorithm.

\begin{figure}[ht]
\fbox{\parbox{\textwidth}{

\noindent
{\bf Volume($K,\eps$).}
\begin{enumerate}
\item Let $K_1 = K$ and $T = \log^* n$
\item For $i=1 \ldots T-1$,
\begin{enumerate}
\item Compute an approximate $\ell$-ellipsoid of $K_i$ using the convex program (\ref{sdp}) to get an approximately
optimal transformation $A_i$ (the corresponding ellipsoid is $A_iB_2^n$).
\item Set
\[
r_{in} =  \frac{\eps \sqrt{n}}{\sqrt{\ln(1/\eps)} C\log^{(i)}(n) \tilde{\ell}_{K_i}(A_i)} \mbox{ and }
r_{out} = \frac{C\sqrt{\ln(1/\eps)} \log^{(i)}(n)\tilde{\ell}_{K_i^*}(A_i^*)}{\eps \sqrt{n}}.
\]
\item Define
\[
K_{i+1} = \mbox{conv}\{K_i \cap r_{out} A_iB_2^n, r_{in} A_iB_2^n \}.
\]
\end{enumerate}
\item Compute the ellipsoid $E = r_{in} A_{T-1} B_2^n$ and
a maximum volume parallelopiped $P$ inscribed in $E$ (via the principal components of $A_{T-1}$).
\item Cover $K_T$ with disjoint copies of $\eps P$. Output $k \vol(P)$ where $k$ is the number of copies used.
\end{enumerate}
}}
\caption{Deterministic Volume Algorithm}
\label{fig:vol-algo}
\end{figure}

\begin{proof}[Proof of Theorem \ref{thm:optimal-vol}]
Let $a_i = \log^{(i)} n$. As in Lemma \ref{lem:it-banach-mazur}, we bound the Banach Mazur via the following recurrence
\begin{align*}
d_{BM}(K_{i+1},B_2^n) \le r^i_{out}/r^i_{in} \le C \frac{\ln(1/\eps)}{\eps^2} (\log^{(i)} (n))^2 (\log d_{BM}(K_{i},B_2^n))^{\frac{3}{2}}.
\end{align*}
From the above recurrence a direct computation reveals that for $\forall~ i \in [T]$,
\[
d_{BM}(K_i,B_2^n) \leq C \frac{\ln(1/\eps)^{5/2}}{\eps^2} (\log^{(i-1)}(n))^{\frac{7}{2}}
\]

We now show that the volumes of the $K_i$ bodies changes very slowly. This will enable us to conclude that the volume of $K_T$ is
very close to the volume of $K$.

By Lemmas \ref{lem:ell-containment}, \ref{lem:ell-approx} and the above bound on $d_{BM}(K_i,B_2^n)$, we have that
\[
r^i_{in}A_iB_2^n \subseteq r^i_{in}\ell_{K_i}(A_i) K_i \subseteq C \frac{\eps \sqrt{n \log d_{BM}(K_i, B_2^n)} }{\sqrt{\ln(1/\eps)}\log^{(i)}(n)} K_i \subseteq C \eps \sqrt{n} K_i
\]
and that
\[
r^i_{out}A_iB_2^n = C \frac{\sqrt{\ln(1/\eps)}\log^{(i)}(n)\tilde{\ell}_{K^*}(A^{-1})}{\eps \sqrt{n}} A_iB_2^n \supseteq C \frac{\ell_{K^*}(A^{-1})}{\eps \sqrt{n}} A_iB_2^n
\supseteq C \frac{1}{\eps \sqrt{n}} K_i \text{.}
\]

\noindent Therefore if $\eps \leq C/\sqrt{n}$, then $K_{i+1} = \conv \set{r^i_{in}A_iB_2^n, K_i \cap r^i_{out}A_iB_2^n} = K_i$.
Since this holds for all $i \in [T-1]$, we get that $K_T = K$ and hence $\vol(K_T) = \vol(K)$.
\vspace{1em}

\noindent Now assume that $\eps \geq C/\sqrt{n}$. Then for $i \in [T-1]$, using Lemmas \ref{lem:covering1} and \ref{lem:Sudakov}, we have,

\begin{eqnarray*}
\vol(K_{i+1}) &\ge& \vol(K_i \cap r_{out} B_2^n)\\
&\ge& \frac{\vol(K_i)}{N(K_i, r_{out}^iB_2^n)} \\
&\ge& e^{-C (\ell_{K_i^*}(A_i^{-1})/r_{out}^i)^2} \vol(K_i)\\
&\ge& e^{-C (\eps^2/\ln(1/\eps)) n\log d_{BM}(K_i,B_2^n)/a_i^2} \vol(K_i) \\
&\ge& e^{-C n \eps / \log^{(i)}(n)} \vol(K_i).
\end{eqnarray*}

From the above, we get that
\[
\frac{\vol(K_T)}{\vol(K)} = \prod_{i=1}^{T-1} \frac{\vol(K_{i+1})}{\vol(K_i)} \geq \prod_{i=1}^{T-1} e^{-C n \eps / \log^{(i)}(n)} \geq e^{-C n \eps}
\]

Next via Lemma \ref{lem:volume1}, the above containment, and Lemma \ref{lem:Sudakov}, we have,
\begin{eqnarray*}
\vol(K_{i+1}) &\le& \vol(\conv{\{K_i, r_{in}B_2^n\}})\\
&\le& C (\eps \sqrt{n}) n N(r_{in}B_2^n, K_i) \vol(K_i)\\
&\le& C (\eps n^{\frac{3}{2}}) e^{C(r_{in}^i\ell_K(A_i))^2} \vol(K_i) \\
&\le& C (\eps n^{\frac{3}{2}}) e^{C(\eps^2/\ln(1/\eps)) n \log d_{BM}(K_i,B_2^n)/a_i^2} \vol(K_i) \\
&\le& C (\eps n^{\frac{3}{2}}) e^{C n\eps / \log^{(i)}(n)} \vol(K_i).
\end{eqnarray*}

From this, we get that
\[
\frac{\vol(K_T)}{\vol(K)} = \prod_{i=1}^{T-1} \frac{\vol(K_{i+1})}{\vol(K_i)} \leq (C\eps n^{\frac{3}{2}})^{\log^*(n)} \prod_{i=1}^{T-1} e^{C n \eps / \log^{(i)}(n)} \leq e^{C n \eps},
\]
where the above holds as long as $\eps = \Omega(\frac{\log n \log^* n}{n})$ (which we have by assumption).

Combining the above inequalities, we get
\[
e^{-C \eps n} \vol(K) \le \vol(K_T) \le e^{C \eps n} \vol(K).
\]

Let $E$ denote the final ellipsoid computed by the algorithm, and let $P$ denote a maximimum volume inscribed
parallelipiped of $E$. By construction of $E$ and $K_T$, we have that $E \subseteq K_T \subseteq
C\frac{\ln(1/\eps)^{5/2}}{\eps^2}E$. Therefore the covering produced is contained in $K_T + \eps P \subseteq K_T + \eps
E \subseteq (1+\eps) K_T$.  Hence the estimate found by the algorithm lies between $\vol(K_T)$ and $\vol((1+\eps)K_T)
= (1+\eps)^n \vol(K_T)$. Thus the overall approximation factor is bounded by $e^{C n \eps}$
as desired.

Next we bound the size of the covering found by the algorithm in Step 4.
Noting that $\vol(E) = 2^{O(n)} \vol(P)$, the size of the covering is bounded by
\[
\frac{\vol(K_T + \eps P)}{\vol(P)} \leq (1+\eps)^n \frac{\vol(K_T)}{\vol(P)} \leq C^n (1+\eps)^n \frac{\vol(K_T)}{\vol(E)} \leq
C^n (1+\eps)^n (\ln(1/\eps)^{5/2}/\eps^2)^n = (1/\eps)^{O(n)} \text{.}
\]
Finally, we describe the enumeration procedure that will ensure that the time bound is $(1/\eps)^{O(n)}$ and
the space used is polynomial in $n$. The number of parallelopipeds enumerated could be as high as $(1/\eps)^{O(n)}$.
However, we do not need to store all the copies that intersect $K$, we only need the number.
To do this using polynomial space, we start with a parallelopiped inside $K$
designated as the {\em root} and fix an order on its axes. For every other parallelopiped
in the axis-aligned tiling, designate its {\em parent} to be an adjacent node closer to the root in Manhattan distance
along the axes of the parallelopiped (i.e., the usual $L_1$ distance for the centers
of the parallelopipeds after transforming parallelopipeds to cuboids), breaking ties using the ordering on coordinates. This ensures that a depth-first traversal
of the tree defined by this structure takes time linear in the number of nodes in the traversal and space
linear in the dimension.  This is a special case
of a more general space-efficient traversal technique studied by Avis and Fukuda \cite{AvisF93}.
\end{proof}

\section{Applications to lattice problems}\label{sec:SVP}

We now consider the consequences of our deterministic M-ellipsoid algorithm for lattice problems, in particular, the shortest vector problem (SVP) and closest vector problem (CVP) in any norm. Dadush et al \cite{DPV-SVP-11} who gave reductions from SVP and CVP any norm to the case of Euclidean norm, or more specifically, to the problem of enumerating all lattice points in an ellipsoid.
This special case was solved in $2^{O(n)}$ time for both SVP and CVP by Micciancio and Voulgaris  \cite{DBLP:conf/stoc/MicciancioV10}, using an approach specific to the Euclidean norm.
The key idea of the reduction was to cover a suitable scaling of the convex body $K$ defining the norm by $2^{O(n)}$ translations of its M-ellipsoid. The scaling $s > 0$ is such that $sK \cap L \neq \emptyset$ and
$\frac{s}{2}K \cap L = \emptyset$. At this scaling, a simple volume argument shows that any translation of $sK$ contains no more than $2^{O(n)}$ lattice points. To enumerate lattice points in $sK$, the idea is to
cover $sK$ using translations of an ellipsoid $E$, and then to enumerate lattice points in each of these ellipsoids.
The number of points in any of the ellipsoids is easily bounded by the number of translations of $sK$ required to cover $E$ times the number of points in any translation of $sK$, i.e., $N(sK, E) N(E,sK)$;
this is precisely the number that is bounded by $2^{O(n)}$ for an M-ellipsoid. Thus, the complexity of the
algorithm is $2^{O(n)}$ {\em plus} the complexity of computing an M-ellipsoid.
The approach for CVP is similar. We now state the reduction precisely.
For a lattice $L$ and convex body $K$ in $\R^n$, let $G(K,L)$ be the
largest number of lattice points contained in any translate of $K$, i.e.,
\begin{equation}
  \label{eq:G-K-L}
  G(K,L) = \max_{x \in \R^{n}} \abs{(K+x) \cap L}.
\end{equation}
The main result of
\cite{DPV-SVP-11} can be stated as follows.
\begin{theorem}\cite{DPV-SVP-11}
  Given any convex body $K \subseteq \R^n$ along with an ellipsoid $E$ of $K$ and any $n$-dimensional
  lattice $L \subseteq \R^n$, the set $K \cap L$ can be computed in
  deterministic time $G(K,L) \cdot N(K,E)N(E,K) \cdot 2^{O(n)}$.
\end{theorem}

For an M-ellipsoid $E$ of $K$, the numbers $N(K,E)$ and $N(E,K)$ are both bounded by $2^{O(n)}$.

From Theorem \ref{thm:det-M-ellipsoid}, we obtain a simple corollary.
\begin{corollary}
  Given any convex body $K \subseteq \R^n$ and any $n$-dimensional
  lattice $L \subseteq \R^n$, the set $K \cap L$ can be computed deterministically in time $G(K,L) \cdot 2^{O(n)}$.
\end{corollary}

For SVP in any norm, as showin in \cite{DPV-SVP-11}, a simple packing argument shows that $G(\lambda_1 K,L) = 2^{O(n)}$, where $\lambda_1 = \inf_L \|x\|_K$,
the length of the shortest nonzero vector in $L$, thus implying Theorem \ref{thm:det-SVP}. Similarly, for CVP in any norm
$G(\gamma \lambda_1 K,L) = (2+\gamma)^{O(n)}$, where $\gamma$ is the ratio between the minimum distance of the query point to the lattice and the length of the shortest nonzero vector; this gives Theorem \ref{thm:det-CVP}. Both these bounds on $G(K,L)$ can be found in \cite{DPV-SVP-11}.

It remains open to solve CVP deterministically in time $2^{O(n)}$ with no assumptions on the minimum distance. Even the special case of CVP under the $L_\infty$ norm is open.\\

\noindent {\bf Acknowledgements.} We are very grateful to Assaf Naor and Grigoris Paouris for helpful pointers and discussions, 
specifically for showing us the proofs of Lemmas \ref{easy-ell-approx} and \ref{lem:ell-approx}.

\bibliographystyle{pnas}
\bibliography{lattices,acg,cg}

\end{document}